\documentclass[aps,twocolumn,showpacs,superscriptaddress]{revtex4-1}
\usepackage{cmap}
\usepackage[T1]{fontenc}

\usepackage[table]{xcolor}
\usepackage{xspace,amsfonts,amsthm,amssymb,amsbsy,graphics,graphicx,bm}
\usepackage[fleqn]{amsmath}
\usepackage{epigraph}
\usepackage{soul} 
\usepackage{epstopdf}
\usepackage{bbding}
\usepackage{mathtools}
\usepackage{pifont}
\usepackage{caption}
\usepackage{subcaption}
\usepackage{verbatim}
\usepackage{dcolumn}
\usepackage{bm}
\usepackage{algcompatible}
\usepackage{newfloat}
\usepackage{dsfont}

\DeclareFloatingEnvironment[
    fileext=loa,
    listname=List of Algorithms,
    name=ALGORITHM,
    placement=tbhp,
]{algorithm}

\usepackage[breaklinks=true]{hyperref}
\usepackage{mathrsfs}

\newtheorem{lemma}{Lemma}

\newtheorem{corollary}{Corollary}

\newtheoremstyle{named}{}{}{\itshape}{}{\bfseries}{}{.5em}{#1\thmnote{ (#3) }}
\theoremstyle{named}

\theoremstyle{named}

\theoremstyle{named}

\hypersetup{
  colorlinks   = true, 
  urlcolor     = blue, 
  linkcolor    = blue, 
  citecolor   = red 
}

\newcommand{\bra}[1]{\left\langle{#1}\right|}
\newcommand{\ket}[1]{\left|{#1}\right\rangle}

\setstcolor{Magenta}

\begin{document}

\title{Quantum and super-quantum enhancements to two-sender, two-receiver channels}

\author{Yihui Quek}
\author{Peter W. Shor}
\email{shor@math.mit.edu}

\affiliation{Massachusetts Institute of Technology, Departments of Physics and Mathematics}

\begin{abstract}
\normalsize
We study the consequences of `super-quantum non-local correlations' as represented by the PR-box model of Popescu and R\"{o}hrlich, and show PR-boxes can enhance the capacity of  noisy interference channels between two senders and two receivers. PR-box correlations violate Bell/CHSH inequalities and are thus stronger -- more non-local -- than quantum mechanics; yet weak enough to respect special relativity in prohibiting faster-than-light communication. Understanding their power will yield insight into the non-locality of quantum mechanics. We exhibit two proof-of-concept channels: first, we show a channel between two sender-receiver pairs where the senders are not allowed to communicate, for which a shared super-quantum bit (a PR-box) allows perfect communication. This feat is not achievable with the best classical (senders share no resources) or quantum entanglement-assisted (senders share entanglement) strategies. Second, we demonstrate a class of channels for which a tunable parameter $\epsilon$ achieves a {\em double} separation of capacities; for some range of $\epsilon$, the super-quantum assisted strategy does better than the entanglement-assisted strategy, which in turn does better than the classical one. 
\end{abstract}

 \date{\today}

\maketitle

\section{Introduction}

Bell's influential paper in 1964 \cite{Bell} brought to light the existence of correlations that can be obtained from bipartite measurements of a quantum state, that cannot be reproduced by a local theory. Quantum mechanics is a \textit{non-local theory} because it is able to predict such correlations, whereas a local theory with spatially separated observers could never do so. Such a local theory would \textit{prohibit} physical measurements (of, say, particle A's spin) in one place from affecting the measurement outcomes of another experimenter (who measures, say, particle B's spin) who is spacelike-separated from the first one, if there is no field between them. Whereas in a non-local theory, to borrow an analogy from Popescu \cite{Popescu}, `moving something here, something else instantaneously wiggles there'. 

Research into this area (see \cite{brunner} for a review) has been motivated by the desire to understand \textit{how} the non-locality of quantum theory gives rise to the advantages of information processing with quantum resources. One of the main results in this research is the famous inequality of Clause, Horne, Shimony and Holt \cite{CHSH}, which bounds the statistics of spatially-separated measurements by two experimenters on a physical state in local hidden-variable (LHV) models. They define a quantity
\begin{equation}\label{S}
S := |\langle A_0 B_0 \rangle+\langle A_0 B_1 \rangle+\langle A_1 B_0 \rangle - \langle A_1 B_1 \rangle | 
\end{equation}
where $A_0$ and $A_1$ are local measurement operators corresponding to spin up and spin down on experimenter A's spin-half particle, and $B_0$ and $B_1$ the analogous measurement operators for Bob, and $\langle . \rangle$ denotes expectation value, and show that for LHV models, 
\begin{equation}\label{CHSH}
S_\text{LHV} \leq 2 
\end{equation}
Since LHV theories must obey the inequality \eqref{CHSH}, while quantum theories, which are non-local, need not, the quantity in \eqref{CHSH} has become a popular metric of the non-locality of a given theory and is sometimes referred to as the `CHSH value'. \textit{Quantum mechanics}(QM), as a non-local theory, is exempt from this bound. Measurements on an entangled state, such as the state $\frac{|0\rangle_A|1\rangle_B - |1\rangle_A|0\rangle_B }{\sqrt{2}}$, can satisfy
\begin{equation} \label{QM-CHSH}
S_\text{QM} \leq 2\sqrt{2}. 
\end{equation}
Tsirelson\cite{Tsirelson} proved that with QM, $2\sqrt{2}$ is the maximal achievable violation of this inequality. But QM must also respect the causality/non-signaling property of special relativity, which prohibits information transfer at a speed faster than light. In fact, out of all our physical theories that are currently in use, QM is special in being non-local and \textit{yet} satisfying the non-signaling constraint: two spacelike-separated observers may influence each other \textit{(non-locality)}, and yet, cannot communicate with each other -- the above-mentioned `influence' must not allow for information transfer \textit{(relativistic causality)}. 

But note that even $S_\text{QM}$ falls short of its algebraic maximum (see Equation \eqref{S}), which is 4. In 1994, Popescu and Rohrlich\cite{PR}, asking `Why isn't quantum theory {\em more} non-local?', proved that it is possible to construct causality-satisfying models that are more non-local than QM. To unify these theories, they proposed an abstraction to represent the probability distribution that they induce on measurement outcomes: a \textit{non-local box}, visualized in figure \ref{PRbox}. This is a bipartite correlated box with two ends, one of which is held by Alice and the other by Bob. Alice inputs $x$ (respectively Bob inputs $y$) and the box outputs $a$ (respectively $b$) according to the probability distribution $P(a,b|x,y)$ (where $x, y, a, b \in \{0,1\}$): 

\begin{equation}\label{PRdist}
P^{PR}(a,b|x,y) \left\{
\begin{array}{ll}
      1/2 & \text{if } a \oplus b=xy \\
      0 & \text{otherwise}  \\
\end{array} 
\right. \end{equation}

\begin{figure}
\centering
\includegraphics[scale=0.7]{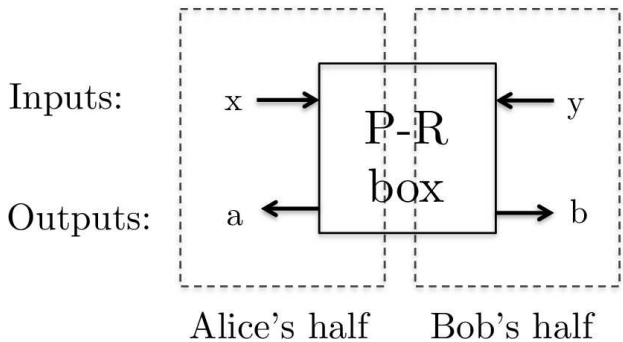}
\caption{\label{PRbox}A PR, or non-local box, whose inputs and outputs are governed by the distribution in Equation \eqref{PRdist}.}
\end{figure}

To calculate the CHSH value of the PR box, we now interpret $A_0, A_1$ (respectively $B_0, B_1$) from Equation \eqref{S} as the expected value of the box's output when Alice (respectively Bob) puts in 0, 1 into her end of the PR-box. This information-theoretic formulation of Alice and Bob's interaction with the theory is completely analogous with our previous language of measurements when construed within the measurement-operator formalism: in making measurements of a two-level system, Alice and Bob apply a set of measurement operators $\{\Pi_0, \Pi_1\}$ corresponding to the two possible outcomes, which correspond exactly to the set of inputs $\{0,1\}$ of both experimenters to the PR-box. 

Thus, with such a PR box, we achieve the following super-quantum correlations:
\begin{equation}
S_\text{QM} = 4, \qquad 
\end{equation}
and we call a theory that predicts the correlations of PR-boxes a \textit{super-quantum} theory. This is one that produces even stronger nonlocal correlations than quantum theory. 

An important implication of this is that \textbf{if they share a PR-box, Alice and Bob could always win the CHSH game}. This is because the condition for outputs $a,b$ to be produced by the PR box is exactly the winning condition of the CHSH game. On the other hand, if Alice and Bob share an entangled pair, they could win the CHSH game with a probability of at most $\cos^2(\frac{\pi}{8})$. This illustrates the super-quantum nature of the PR-box. 

We summarize the theories under consideration in terms of their locality properties (as measured by their CHSH value) in Figure \ref{hierarchy}. We also refer the reader to \cite{PL} for a comprehensive review of PR-boxes and non-local correlations.

\begin{figure}
\centering
\includegraphics[scale=0.4]{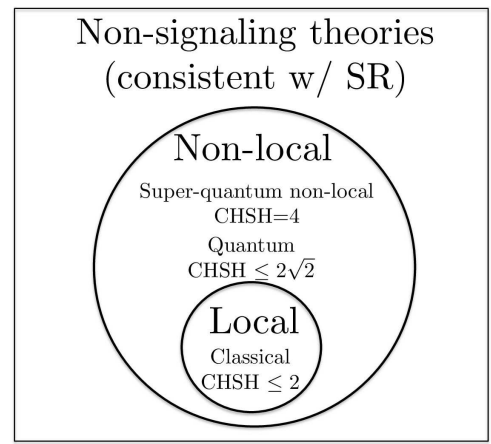}
\caption{\label{hierarchy} Types of theories grouped by their locality properties (they must all not permit space-like separated observers to communicate and hence all fall under the banner of non-signaling)}
\end{figure}

A major push of quantum information research has been to devise strategies that utilize quantum properties, such as entanglement, to aid communication tasks -- quantum key distribution, quantum bit commitment and so on. This begs the question: could we use the maximally non-local correlations of \textit{super-quantum} theories as a resource, and what tasks would they facilitate? 

Previously, PR-boxes had been shown to allow Alice and Bob to perform any two-party distributed computation by transmitting only a single bit of information \cite{vandam}, as well as the cryptographic primitives of unconditionally-secure bit-commitment and oblivious transfer \cite{nonlocalcrypto}. This paper is the first survey of how super-quantum assistance could enhance communication over an interference channel. 

It is organized as follows: In Section \ref{sec: 3}, we introduce notation for the two-sender, two-receiver interference channel, as well as the information quantity we optimize. In Section \ref{sec: 4}, we present our original result of a two-sender, two-receiver interference channel over which communication is more efficient with the aid of a PR-box, than with entanglement and/or a classical strategy. In Section \ref{sec: 5}, we present a variant of the above; a class of erasure channels characterized by a tunable parameter $\epsilon$, whose capacities show a strict separation given these three classes of resources (classical, quantum-assisted and PR-box assisted). We finally conclude with a summary of results and suggestions for future research in Section \ref{sec: 6}.

\section{Notation} \label{sec: 3}
In the following sections, we will exhibit several two-sender, two-receiver channels that demonstrate capacity separations. We use the Shannon model of channel communication \cite{Shannon1} to describe these channels, for which we follow the notation of \cite{NetworkIT} (in turn based on \cite{sato}). The basic model of a two sender-receiver pair channel is depicted in figure \ref{interference}.

\begin{figure*}[!htbp]
\includegraphics[scale=0.5]{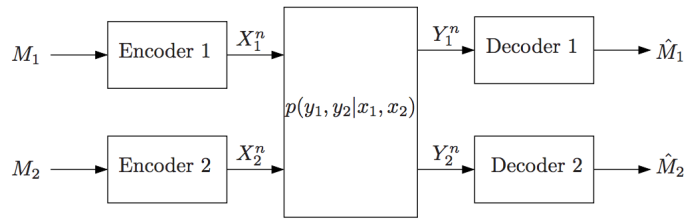}
\caption{\label{interference} General model of a two sender-receiver pair communication system. Figure taken from \cite{NetworkIT}.}
\end{figure*}
Such a channel is denoted $(\mathcal X_1 \times \mathcal X_2, p(y_1,y_2|x_1,x_2), \mathcal Y_1 \times \mathcal Y_2)$. A  $(2^{nR_1},2^{nR_2},n)$ code for this channel consists of:
\begin{itemize}
\item Two message sets $[1:2^{nR_1}]$ and $[1:2^{nR_2}]$
\item Two encoders, where encoder 1 assigns a codeword $x_1^n(m_1)$ to each message $m_1 \in [1:2^{nR_1}]$ (respectively encoder 2 assigns $x_2^n(m_2)$ for $m_2 \in [1:2^{nR_2}]$).
\item Two decoders, where decoder 1 uses a {\em decoding rule} to assign an estimate $\hat{m_1}$ or an error message $e$ to each received sequence $y_1^n$, and decoder 2 does the same (ie. assigns $\hat{m_2}$ or $e$).
\end{itemize}
A rate pair $(R_1,R_2)$ is said to be achievable for this channel if there exist a sequence of $(2^{nR_1},2^{nR_2},n)$ codes such that $\underset{n\rightarrow \infty}{\lim} [ P_e^{(n)} \equiv P\{(\hat{M_1},\hat{M_2}) \neq (M_1,M_2)\}] = 0.$ For our channels, we will be concerned with their sum-capacity $C_\text{sum}$ for classical information, which is the maximum, over all coding strategies, of the sum of the rates for each sender-receiver pair. That is, 
\begin{equation} \label{sumcap}
C_\text{sum} = \max_\text{coding strategy} \left( R_1+R_2 \right) 
\end{equation}
Note that $C_\text{sum} \neq \underset{p(x_1)}{\max} R_1+ \underset{p(x_2)}{\max} R_2$ in general, because the rates must be attainable \textit{simultaneously}. Whenever we speak of the `capacity' of a channel, we shall refer to this information capacity.

\section{Channel I} \label{sec: 4}

In 2005, Cerf, Gisin, Massar and Popescu demonstrated a sense in which super-quantum non-locality encompasses quantum non-locality -- they showed that a PR box could simulate the correlations obtained from any bipartite measurement of a maximally entangled pair of qubits without communication\cite{CGMP}. The reverse direction of simulation is impossible because PR-box correlations are more non-local than entanglement. Therefore, one expects that any communication task which is made more efficient with the aid of entanglement, could potentially benefit {\em even more} from PR-boxes. Table \ref{1} represents a channel that demonstrates just such a non-locality separation. In what follows, `classical' strategies are those where senders are allowed to share no communication but may discuss a strategy before-hand, and `entanglement-assisted' (alternatively, `quantum-assisted)  strategies imply strategies where the senders are allowed to share $2\times n$ quantum entanglement -- that is, a bipartite quantum state where each half is an $n$-level system represented as a $n$-dimensional Hilbert space. 

\begin{table}[!htbp]
\centering
  \begin{tabular}{ |c || c | c | c | c ||}
  \hline
 $X_1$\textbackslash $X_2$  & \textbf{00} & \textbf{01} & \textbf{10} & \textbf{11} \\ \hline \hline
\textbf{00} & 00   &  11 & 01  & 10 \\ \hline
\textbf{01} & 11  & 00  & 10 & 01 \\ \hline
\textbf{10} & 10 &  01 & 00 & 11  \\ \hline
\textbf{11} & 01  &  10 & 11  & 00 \\ \hline
\hline
  \end{tabular}
\caption{Channel I: The senders each send two-bit codewords, $X_1$ and $X_2$ (codeword choices are in bold, on the axes), and the two-bit entries in the table ($Y_1Y_2$) correspond to the channel outputs; one bit goes to each receiver. Thus, if Sender 1 sends $01$ and Sender 2 sends $10$, Receiver 1 gets $1$ and Receiver 2 gets $0$. \label{1}}
  \end{table}
  
The notation for this channel (which we shall call Channel I) is as follows: the senders and receivers shall be denoted by $A_i$ and $B_i$; the bits they handle shall be denoted $X_i$ (2 bit message that $A_i$ inputs to the channel) and $Y_i$ (1 bit message that $B_i$ receives from the channel). To prevent confusion, we will try not to use $A/B$ simultaneously with $X/Y$, unless it is necessary to make such a distinction. 

On each use of Channel I, the senders send \textbf{two bits} out of the alphabet $\{00,01,10,11\}$ and the channel outputs \textbf{one bit} to each receiver. Table \ref{1} shows the output pairs that correspond to each input pair.  

By definition, the maximum possible sum-capacity of Channel I (over all classes of resources) is 2: the two receivers each receive one bit. In fact, $C_\text{sum}=2$ only if there exists a strategy where the receiver always decodes the sender's bit perfectly. In fact, it will turn out we fall far short of this maximum if the senders are restricted to using a purely classical probabilistic strategy; in that case the capacity is 1. We now show that this channel demonstrates the capacity separations
\[\textbf{C}_\textbf{classical}, \,\, \textbf{C}_\textbf{quantum} <\textbf{C}_\textbf{super-quantum}.\]

\subsection{Capacity of Channel I with a classical strategy}
Let us build up our intuition about Channel I to understand why the classical strategy capacity should be so small. Channel I takes {\em two}-bit inputs but outputs only one bit to each receiver, so if the senders can ultimately communicate only one bit, the second bit seems redundant. Might the redundancy improve communication? We could note the following:
\begin{itemize}
\item Consider a strategy where each sender sends codewords according to a uniform probability distribution over the entire input alphabet, for both senders. Taking the marginal probability distribution for the first pair (over the second pair) results in the binary symmetric channel of Table \ref{unif}. It is the same for the other pair. This channel has a bit-flip probability $p=0.5$. Since the capacity of the binary symmetric channel is $1-H(p)$, the best possible joint rate with this strategy is 0. 
\begin{table}[!htbp]
  \begin{tabular}{ c || c | c ||}
 $X_1$\textbackslash $Y_1$  & \textbf{0} & \textbf{1} \\ \hline \hline
\textbf{00} & $\Pr = 0.5$  & $\Pr = 0.5$ \\\hline
\textbf{01} & $0.5$ & $0.5$  \\ \hline
\textbf{10} & $0.5$ & $0.5$  \\ \hline
\textbf{11} & $0.5$ & $0.5$  \\ \hline
\hline
  \end{tabular} \quad
\caption{\label{unif} A uniform probability distribution results in a perfectly randomizing channel, evident from taking the marginal probability distributions for one sender-receiver pair (in this case, the first).}
  \end{table}
\item The following coding strategy gives a joint rate of 1, and therefore 1 is an inner bound on the sum capacity: $A_2$ always sends $00$ while $A_1$ encodes message bit 0 as $00$ and message bit 1 as $01$; then $B_1$ receives exactly the bit that $A_1$ intended to send. So the first sender-pair always communicates perfectly at the expense of the second pair. 
\end{itemize}
The reader should persuade herself that other simple strategies such as reducing the size of either sender's alphabet will not achieve perfect coding either. In fact, as  Lemma \ref{joint} shows, it is not even possible to do {\em better} than $R_1+R_2 = 1$. 

\begin{lemma}[Classical capacity of Channel I] \label{joint}
If the senders are limited to a classical (at most probabilistic) strategy with no aid from communication, entanglement or PR boxes, on the given channel the sum-capacity is strictly outer-bounded:
\begin{equation}\label{1.5}
R_1 + R_2 < 2.
\end{equation}
In fact, we may show computationally that 
\begin{equation} \label{1.7}
C_\text{classical} = 1. 
\end{equation}
\end{lemma}
\begin{proof}\label{proof} Here we sketch a proof for Equation \eqref{1.5}. We show that $R_1 := I(X_1:Y_1) =1$ implies $R_2:= I(X_2:Y_2) < 1$. 

Suppose $I(X_1:Y_1) =1$. Using the chain rule for mutual information shows that $I(X_2:Y_1|X_1) =0$.
\begin{equation}
I(X_1:Y_1) = 1 =  \underbracket{I(X_1,X_2:Y_1)}_{\text{takes on maximal value, 1}}
 - \underbracket{I(X_2:Y_1|X_1)}_{=0}
\end{equation}
Information-theoretically, the condition $I(X_2:Y_1|X_1) =0$ means that the first receiver's bit, $Y_1$, cannot possibly distinguish between the possibilities for the second sender's 2-bit message, $X_2$, for every choice of $X_1$ -- and this is a restriction on what the the second sender's alphabet set could be. Consider the first row of Table \ref{1}. The restriction says that \textbf{if Sender 1 sends $00$ on a particular channel use, then there are only two possible non-trivial choices for Sender 2's alphabet}: a uniform probability distribution over $\{00,10\}$ (both resulting in the output $Y_1 = 0$), OR a uniform probability distribution over $\{01,11\}$ (both resulting in the output $Y_1 = 1$). 

We similarly analyze the cases when Sender 1 sends $01,10$ or $11$. The conclusion is that one of the following must hold (otherwise the restriction is never met):
\begin{enumerate}
\item Sender 1's alphabet is some subset of $\{00,01\}$; Sender 2's alphabet is either $\{00,10\}$ or $\{01,11\}$.
\item Sender 1's alphabet is some subset of $\{10,11\}$; Sender 2's alphabet is either $\{00,11\}$ or $\{10,01\}$.
\end{enumerate}
Since the two senders are not allowed to communicate during the sending of the messages, they must choose an alphabet at the start and stick to it. Consequently, only one of these four cases can hold, and bearing in mind the other restriction that our coding strategy must fulfil the condition $I(X_1:Y_1) = 1$, we may show that $R_2<1$ for all of them. For an example of this analysis, refer to Appendix \ref{app0}. \end{proof}

But it is still possible that if one of the sender-receiver pairs is willing to accept a sub-optimal (less than 1) rate, the other pair could attain a high rate such that $R_1+R_2>1$. To show that this never happens, we ran an algorithm based on modified gradient descent. This algorithm is given in pseudocode here (Algorithm \ref{MGD}). The inputs to the algorithm are two vectors $\vec{x_1} := (a_1,b_1,c_1,d_1), \vec{x_2} := (a_2,b_2,c_2,d_2)$, such that the square of the entries in the first vector $\{a_1^2,b_1^2,c_1^2,d_1^2\}$ represents the probabilities of Sender 1 sending $\{00,01,10,11\}$ respectively, and correspondingly $\{a_2^2,b_2^2,c_2^2,d_2^2\}$ for Sender 2. The modification to the usual gradient descent algorithm was to respect the constraints \[a_1^2+b_1^2+c_1^2+d_1^2 =1 \,\,; \,\, a_2^2+b_2^2+c_2^2+d_2^2 =1 .\] To do this, we treated the problem of simultaneous gradient descent where the component vectors had to lie on two 4-D unit spheres. After running gradient descent $10000$ times with a $tol$ set to $1e-6$ and never observing a value of the joint rate above $1$, we concluded that the joint rate is, indeed, upper bounded by 1. Equation \eqref{1.7} follows.

\begin{algorithm}{Modified-Gradient-Descent($\vec{x}$)} 
\caption{\label{MGD}Finds the maximum value of the function $I(X_1:Y_1) + I(X_2:Y_2)$ over all input distributions}
\begin{algorithmic}
\State {$f(\vec{x_1},\vec{x_2}) := -I(X_1;Y_1) - I(X_2;Y_2)$}  (Objective function)
\State {$\vec{g_1} := \vec{\nabla}_{x_1} f$ ; $\vec{g_2} := \vec{\nabla}_{x_2} f$ }  
\State {Initialize $x_1, x_2, tol, maxiter$}
	\WHILE {$iter<maxiter \, \, \text{and} \,\, dx>tol$}
		\STATE {Evaluate $\vec{g_1}(\vec{x_1},\vec{x_2})$ ; $\vec{g_2}(\vec{x_1},\vec{x_2})$}
		\STATE {$\vec{h_1} \leftarrow \vec{g_1}-(\vec{g_1}\cdot \vec{x_1}) \vec{x_1}$ ; $\vec{h_2} \leftarrow \vec{g_2}-(\vec{g_2}\cdot \vec{x_2}) \vec{x_2}$ }
		\STATE{$\alpha_1 \leftarrow \frac{h_1^2}{h_1^2+h_2^2}$ ; $\alpha_2 \leftarrow \frac{h_2^2}{h_1^2+h_2^2}$  }
		\STATE {$\vec{n_1} \leftarrow \frac{\vec{h_1}}{|\vec{h_1}|}$ ; $\vec{n_2} \leftarrow \frac{\vec{h_2}}{|\vec{h_2}|}$ }
		\STATE {$\phi^{'} \leftarrow \arg \underset{\phi}{\min} \,\,f(\cos(\alpha_1 \phi)\vec{x_1} + \sin(\alpha_1 \phi) \vec{n_1}, \cos(\alpha_2 \phi)\vec{x_2} + \sin(\alpha_2 \phi) \vec{n_2})$}
		\STATE {$\vec{x_1^{'}} \leftarrow \cos(\alpha_1 \phi^{'})\vec{x_1} + \sin(\alpha_1 \phi^{'}) \vec{n_1}$; $\vec{x_2^{'}} \leftarrow \cos(\alpha_2 \phi^{'})\vec{x_2} + \sin(\alpha_2 \phi^{'}) \vec{n_2}$} 
		\STATE {$dx \leftarrow \sqrt{x_1^{'2}-x_1^2 + x_2^{'2}-x_2^2}$}
		\STATE {$iter += 1$}
	\ENDWHILE
\end{algorithmic}
\end{algorithm}

\subsection{Capacity of Channel I with super-quantum assistance}

We introduce the notion of super-quantum assisted capacity with a thought experiment: supposing that the two senders may coordinate their input alphabets in real-time, perhaps by using a non-classical resource. If we want both pairs to communicate perfectly, that is $I(X_1:Y_1) =1$ AND $I(X_2:Y_2) =1$, this imposes 4 conditions on the actual encodings that go into the channel:
\begin{enumerate}
\item If $X_1\in \{00,01\}$, either $X_2 \in \{00,10\}$ or $X_2 \in \{01,11\}$.
\item If $X_1 \in \{10,11\}$, either $X_2 \in \{01,10\}$ or $X_2 \in \{00,11\}$. 
\item If $X_2\in \{00,01\}$, either $X_1 \in \{00,10\}$ or $X_1 \in \{01,11\}$.
\item If $X_2 \in \{10,11\}$, either $X_1 \in \{01,10\}$ or $X_1 \in \{00,11\}$. 
\end{enumerate}
That, is, only the shaded outputs in either the left or the right subtables of Table \ref{allowed} could be produced. Obviously, this is not a set that can be produced with only classical resources. Lemma \ref{superquant} states that it is possible with a PR-box.

\begin{table}
\begin{tabular}{ll}
  \begin{tabular}{ |c || c | c | c | c ||} \hline
 $X_1$\textbackslash $X_2$  & \textbf{00} & \textbf{01} & \textbf{10} & \textbf{11} \\ \hline \hline
 \textbf{00} & 00\cellcolor{lightgray} & 11 & 01\cellcolor{lightgray}& 10 \nonumber \\ \hline
 \textbf{01} & 11 & 00\cellcolor{lightgray} & 10 & 01\cellcolor{lightgray} \nonumber \\ \hline
 \textbf{10} & 10\cellcolor{lightgray} & 01 & 00 & 11\cellcolor{lightgray} \nonumber \\ \hline
 \textbf{11} & 01 & 10\cellcolor{lightgray} & 11\cellcolor{lightgray}& 00 \nonumber \\ \hline
  \end{tabular} 
&
    \begin{tabular}{ |c || c | c | c | c ||}\hline
 $X_1$\textbackslash $X_2$  & \textbf{00} & \textbf{01} & \textbf{10} & \textbf{11} \\ \hline \hline
 \textbf{01} & 00& 11\cellcolor{lightgray} &01 & 10\cellcolor{lightgray} \nonumber \\ \hline
 \textbf{00} & 11\cellcolor{lightgray} & 00 & 10\cellcolor{lightgray}& 01 \nonumber \\ \hline
 \textbf{11} & 10 & 01 \cellcolor{lightgray} & 00\cellcolor{lightgray}& 11 \nonumber \\ \hline
  \textbf{10} &01 \cellcolor{lightgray} & 10 & 11 & 00\cellcolor{lightgray} \nonumber \\ \hline
  \end{tabular}
\end{tabular}
 \caption{\label{allowed} Hypothetically, the demand that perfect coding happen requires that only the shaded outputs be produced by the channel. Only these two coding strategies will allow both $I(X_1:Y_1) = 1$ and $I(X_2:Y_2) = 1$. Returning back to the classical realm, since the senders cannot communicate with each other, they cannot coordinate their inputs so as to only produce the shaded outputs, so perfect coding is not possible classically. But if they share a PR-box, they can. Our super-quantum strategy achieves exactly the left-hand-side set of outputs.}
\end{table}

\begin{lemma}[Capacity of Channel I with super-quantum resources]\label{superquant}
If the senders are allowed to share a PR-box, the capacity of the given channel is exactly 2. This is the algebraically maximal sum-capacity of the channel. 
\end{lemma}

We have all but spelled out our super-quantum strategy. In this strategy, the senders can communicate only one bit $m$. They encode this bit into a two-bit codeword by concatenating it with the single-bit output of the PR box which results from feeding $m$ into their respective sides of the box. That is, Sender 1 sends $X_1= ``m_1\,a"$ and Sender 2 sends $X_2= ``m_2\,b"$ where $a,b$ are the outputs of the PR box. This strategy guarantees $a \oplus b= m_1 m_2$. The possible sets of encoded channel inputs produced by this strategy are listed in Table \ref{wincond}. Comparing that to Table \ref{allowed} reveals that the resulting encoded message pairs are special for our channel: {\em they are exactly the combinations whereby Receiver 1 and Receiver 2 respectively receive the original 1-bit messages that Sender 1 and Sender 2 intended to send}. Hence, this super-quantum strategy enables perfect message transmission.
 
\begin{table*}[!htbp]
\centering
\begin{tabular}{ |c|c|c|} 
 \hline
 $(m_1,m_2)$ (PR-box input) & $(a,b)$ (PR-box output) & Encoding (Sender 1, Sender 2) \\ \hline 
 (0,0) & (1,1) or (0,0) & (01,01) or (00,00) \\
 (0,1) & (1,1) or (0,0) & (01,11) or (00,10) \\
(1,0) &  (1,1) or (0,0) & (11,01) or (10,00)  \\
(1,1) &  (0,1) or (1,0) & (10,11) or (11,10) \\ 
 \hline
\end{tabular} 
\caption{\label{wincond} The rightmost column shows all possible combinations of the two senders' inputs to the channel using the encoding strategy described above: each sender's PR-box output (either $a$ or $b$) is concatenated with her input $m_i$.}
\end{table*}

\begin{corollary}[Capacity of Channel I with 1 bit of communication between senders] \label{1bitcom}
If senders are allowed to share one bit of communication, they will achieve a joint rate of 2.
\end{corollary}
\begin{proof}
This strategy follows straightforwardly from the technique described in the proof of the previous lemma. This time, Sender 1 encodes her message bit by duplicating it. She then uses her one bit of communication by sending this message bit to Sender 2. Sender 2 uses this knowledge to replicate the action of the PR-box to pad his own one-bit message, and Table \ref{wincond} shows that this is always possible. Since this strategy achieves (deterministically) exactly the same input sets as the PR-box assisted strategy described above, it too achieves a joint rate of 2. 
\end{proof}

\subsection{Capacity of Channel I with quantum assistance}

We saw that Alice and Bob can coordinate their inputs using a non-classical resource to achieve perfect coding. Does a quantum resource suffice, or only a super-quantum one?

\begin{lemma}[Capacity of Channel I if senders share entanglement]\label{lemquan}
If the senders are allowed to share an entangled quantum state $\ket{\Phi}$ of dimension $2 \times n$, \[C_\text{sum}<2.\]
\end{lemma}

\begin{proof}
For this proof, we borrow notation from \cite{BMT}. Let $\mathcal P$ denote the set of all POVMs acting on a single qubit, and $O_{\mathcal P}$ denote the set of all outcomes for the POVM $\mathcal P$. Let $m_i, X_i, Y_i$ denote the message bits, encoded message bits (input to channel) and channel output bits respectively, where the subscript $i$ denotes the respective sender-receiver pair.The two senders share an entangled state $\ket{\Phi}$.  

\begin{align}
\mathcal C_1 \, : \, m_1 \rightarrow \mathcal P_1 & \quad \mathcal C_2 \, : \, m_2 \rightarrow \mathcal P_2  \nonumber  \\
\mathcal E_1 \, : \, m_1 \times O_{\mathcal P_1\ket{\Phi}} \rightarrow X_1 & \quad \mathcal E_2 \, : \, m_2 \times O_{\mathcal P_2\ket{\Phi}} \rightarrow X_2  \nonumber \\
 \text{Channel}: (X_1,X_2) &\rightarrow (Y_1, Y_2) \nonumber\\
 \mathcal D: (Y_1, Y_2) &\rightarrow (\hat{m_1},\hat{m_2})  \label{mmodel} 
\end{align}

Any quantum strategy for communication can be mathematically represented as four consecutive mappings $(\mathcal C_i , \mathcal E_i, \text{Channel}, \mathcal D)$. The senders independently choose a POVM $(\mathcal C_i)$ depending on their message bit $m_i$, apply that POVM to their share of the entangled state, and apply an encoding function $(\mathcal E_i)$ that maps the measurement outcome of the POVM to a 2-bit input to the channel. These bits go through the channel and the output of the channel is decoded $(\mathcal D)$ by the two receivers. This process is illustrated in Figure \ref{gen}.  

This is indeed the most general form of a quantum communications strategy; Naimark's theorem guarantees that a POVM is mathematically equivalent to a general measurement, and the most general decoder looks at both the channel outputs (including as a special case a restricted decoding strategy where receivers do not communicate). This model (and the proof it inspires) is in very much the same spirit as the model in \cite{BMT}, which was used to prove a similar result for two-player pseudo-telepathy games. 

\begin{figure}[!htbp]
\includegraphics[scale=0.35]{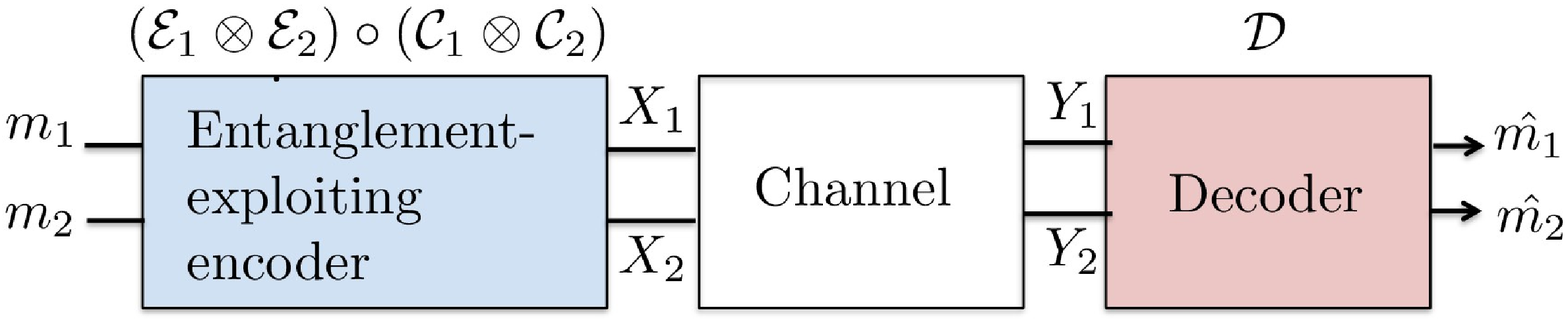}
\caption{\label{gen}Model of a quantum communication system over this channel.}
\end{figure}

Our goal is to show that that if there exists a quantum strategy that achieves rate 2 (ie. perfect coding), there is a classical strategy that achieves the same rate. But, since there is {\em not} a classical strategy that achieves perfect coding, there cannot be a quantum one. 

We assume that any decoding strategy depends deterministically, and solely, on the bits that the receivers receive. That is, every time a particular $(Y_1, Y_2)$ is received, the decoding step infers a fixed, corresponding, $\hat{m_1}$ and $\hat{m_2}$. Demanding a rate of 2 rules out any probabilistic decoding strategy, so the inferred $\hat{m_1}$ and $\hat{m_2}$ have to be the right ones. The question is now whether there exist functions $(\mathcal E_1 \otimes \mathcal E_2) \circ (\mathcal C_1 \otimes \mathcal C_2)$ such that the overall map from $m_1 \times m_2$ to $(X_1 \times X_2) / (Y_1 \times Y_2) $ is injective. This means we can group the 16 options for $(X_1, X_2)$ based on the resulting $(Y_1,Y_2)$ and stipulate that the entanglement-assisted $(\mathcal E_1 \otimes \mathcal E_2) \circ (\mathcal C_1 \otimes \mathcal C_2)$ must achieve the following map:

\begin{table}[!htbp]
\centering
  \begin{tabular}{| c || c |} \hline
$\mathbf{(m_1, m_2)}$ & $\mathbf{(X_1, X_2)}$ \\ \hline
$(m_1, m_2)_a$ & $(00,00), (01,01), (10,10), (11,11)$ \\ \hline
$(m_1, m_2)_b$ & $(00,01), (01,00), (10,11), (11,10)$ \\ \hline
$(m_1, m_2)_c$ & $(00,10), (01,11), (10,01), (11,00)$ \\ \hline
$(m_1, m_2)_d$ & $(00,11), (01,10), (10,00), (11,01)$ \\ \hline
\hline
 \end{tabular}
\caption{\label{map} Entanglement-assisted map between message bits and their encoding. $(m_1, m_2)_a, (m_1, m_2)_b, (m_1, m_2)_c, (m_1, m_2)_d$ must correspond to some permutation of the message set $\{(0,0), (0,1),(1,0),(1,1)\}$.}
  \end{table}

Therefore, all we are asking of our copycat classical strategy is that, for any combination of message bits, it should encode them as some subset of the allowed encodings in the right column of the corresponding row -- since this suffices for perfect decoding. Note that the assumption that our quantum strategy is {\em perfect} is key -- our classical strategy only needs to never produce an illegal output, even though some legal outputs may never occur. \textbf{It turns out that it is entirely possible to devise a classical strategy that never produces an output that is illicit from a POVM}, and this is proved in Appendix \ref{app1}.
\end{proof}

\section{Channel II} \label{sec: 5}

In this section, we present a class of related channels to Channel I that displays yet stronger capacity separations: 
\[\textbf{C}_\textbf{classical} < \textbf{C}_\textbf{quantum} < \textbf{C}_\textbf{super-quantum}.\] 

To get Channel II, we modify Channel I by allowing now two types of outputs. Consider the cells in Table \ref{2}. The cells with $ee$ are outputs that always get erased. All other cells are outputs that are {\em erased} with probability $1-\epsilon$, but with probability $\epsilon$ output the two bits stated. We will see later that the parameter $\epsilon$ can be tuned to change the magnitude of the capacity separations. We prove the desired inequalities for this channel when $\epsilon$ is taken to be small.
\begin{table}[!htbp]
\centering
  \begin{tabular}{ |c || c | c | c | c ||} \hline
 $X_1$\textbackslash $X_2$  & \textbf{00} & \textbf{01} & \textbf{10} & \textbf{11} \\ \hline \hline
\textbf{00} & 00/ee   &  ee & 01/ee  & ee \\ \hline
\textbf{01} & ee  & 00/ee  & ee & 01/ee \\ \hline
\textbf{10} & 10/ee & ee & ee & 11/ee  \\ \hline
\textbf{11} & ee &  10/ee & 11/ee  & ee \\ \hline
\hline
  \end{tabular}
\caption{\label{2}Channel II: a variation on Channel I in which the channel outputs not corresponding to the PR-box-encoded joint inputs are erased with probability 1, and the channel outputs corresponding to the PR-box-encoded joint inputs are erased with probability $1-\epsilon$. Erased bits are denoted by `e'.}
 \end{table}
\subsection{Super-quantum and entanglement-assisted capacities of Channel II}
\begin{lemma}[Capacity of Channel II with super-quantum resources]
There exists a super-quantum-assisted strategy on Channel II that achieves $R_1+R_2 = 2\epsilon$.
\end{lemma}
\begin{proof}
Encoding proceeds exactly as in in the previous section. Why this works is best visualized by comparing our channel in Table \ref{2} to the set of encoded messages produced by the PR-box strategy from the previous section, summarized in the left half of Table \ref{allowed} -- {\em the encoding only produces the channel inputs whose outputs are erased with probability $1-\epsilon$ by Channel II}. Since preserved outputs contain exactly the first bits of each sender's message, they are perfectly decoded by each receiver. This therefore amounts to a binary erasure channel for each sender-receiver pair with erasure parameter $1-\epsilon$. This gives a joint rate of $2\epsilon$.
\end{proof}

This intuition is this: any classical choice of input alphabets for the two senders results in at least one combination of inputs that is always erased by the channel. Using a PR-box helps us avoid these `bad' input combinations, and using entanglement helps us avoid them with probability $\cos^2(\frac{\pi}{8}) \approx 0.854$, as we will see next.

\begin{lemma}[Achievable rate with senders sharing entanglement] \label{quant}
There exists an entanglement-assisted strategy on Channel II that achieves $R_1+R_2 = \left[2 \cos^2(\frac{\pi}{8})\right]\epsilon$.
\end{lemma}
\begin{proof}
We will describe such a strategy. The encoding step is a simple extension of the previous one: in place of the PR-box, let the two senders share the CHSH entangled pair, $|\Psi\rangle = \frac{\ket{00}+\ket{11}}{\sqrt{2}}$. This is the same state that they can use to win the CHSH game with higher-than-classical probability. The essence of the strategy is that they play a CHSH game to communicate. 
Recall that the winning condition of the CHSH game is that 
\begin{equation}\label{eqn:wincond}
a \oplus b= r \wedge s
\end{equation}
where $r := \text{player 1's question}$, $s :=\text{player 2's question}$, $a := \text{player 1's response}$, $b: =\text{player 2's response}$. With a shared Bell state, the two players can perform measurements on their state in such a way that their question and response bits fulfill Equation \eqref{eqn:wincond} with probability $\cos^2(\frac{\pi}{8})$. But observe that this is exactly the equation that {\em always} holds true for all licit input-output pairs (inputs: $r, s$, outputs: $a,b$) from a PR-box. Therefore, instead of concatenating the PR-box output with their message bit, the senders now concatenate $a$ or $b$ with their message bit, where $a$ and $b$ are obtained by measurements on their shared entangled state. That is, $a$ and $b$ are their `response' bit in the CHSH game if their desired message had been their `question' bit from the referee. 

This encoding strategy allows for pretty-good communication. We may observe that we obtain a `good' encoding (one that lands on the double-valued cells in Table \ref{2}) with probability $\cos^2\left(\frac{\pi}{8}\right) \approx 0.854$; we obtain a `bad' encoding (one that lands on the single-valued cells, thus always gets erased) with probability $\sin^2\left(\frac{\pi}{8}\right) \approx 0.147$. Hence, each sender gets his input bit erased with probability $\alpha = \sin^2(\frac{\pi}{8}) + (1-\epsilon) \cos^2(\frac{\pi}{8})$, and transmitted perfectly with probability $\epsilon\cos^2(\frac{\pi}{8})$. This amounts to each sender-receiver pair experiencing a binary erasure channel with erasure probability $\alpha$. Since the capacity of a binary erasure channel is $1-\alpha$, the joint rate achieved by such a strategy is $2(1-\alpha) \approx 1.707 \epsilon $.
\end{proof}

\subsection{Capacity separations for Channel II}
Finally, we reach the capstone lemmas of this section:
\begin{lemma}[Classical vs. quantum capacities of Channel II]
For sufficiently small $ \epsilon, \textbf{C$^{(\epsilon)}$}_\textbf{classical} < \textbf{C$^{(\epsilon)}$}_\textbf{quantum}$.
\end{lemma}

\begin{proof}
In Appendix \ref{appB} we prove a lemma that upper-bounds $\textbf{C$^{(p)}$}_\textbf{classical}$ by $1.255\epsilon+O(\epsilon^2)$. This proof rests on the following Lemma. 

\begin{lemma} \label{lem}
Suppose the inputs to a channel are five symbols, $1$, $2$, $3$, $4$, and
$?$. The first four symbols are replaced by $?$ with probability
$1-\epsilon$ and transmitted intact with probability $\epsilon$, and the
last symbol is always sent as $?$. Furthermore, suppose that these
symbols must be sent with probabilities $p_1$, $p_2$, $p_3$, $p_4$, and
$p_?$, with these probabilities adding up to 1. Then the capacity of this
channel is
\begin{equation} \label{exp}
- \epsilon \left(p_1 \log p_1 + p_2 \log p_2 + p_3 \log p_3 
+ p_4 \log p_4 \right) + O(\epsilon^2)
\end{equation}
\end{lemma}

(This lemma is perfectly mappable to our problem; we need only consider the $p_i$s to be the probabilities of an $xx/ee$ state being sent, where $xx$ is one of $\{00,01,10,11\}$. That is, eventually we wish to make the replacement:
\begin{align} \label{sub}
p_1 &\leftarrow  pq\alpha \nonumber \\
p_2 &\leftarrow  p(1-q)\beta \nonumber \\
p_3 &\leftarrow  (1-p)q\gamma \nonumber \\
p_4 &\leftarrow  (1-p)(1-q)\delta )
\end{align}

{\bf Proof}:
The probability of the output symbol $?$ is \[p_? + (1-\epsilon) (p_1+p_2+p_3+p_4) = p_? + (1-\epsilon) (1-p_?) = 1-\epsilon+ \epsilon p_?. \] 
We simply plug into Shannon's formula (X = channel input, Y = channel output):
\[I(X;Y) = H(Y) - H(Y|X), \] where
\begin{align*}
H(Y|X) & = \Sigma_\text{i=1}^{4} p_i H(\epsilon) \\
H(Y) & = - \Sigma_\text{i=1}^{4} \epsilon p_i \log(\epsilon p_i)\\
  &- (\epsilon p_? + 1-\epsilon) \log (\epsilon p_? +1-\epsilon).  
\end{align*}
The second term of these goes to zero as $\epsilon \rightarrow 0$, so we ignore it henceforth. 
\begin{align*}
& H(Y)-H(Y|X) \\
&=  \Sigma_\text{i=1}^{4} -p_i [\epsilon \log(\epsilon p_i) + H(\epsilon)] \\
&=  \Sigma_\text{i=1}^{4} - p_i [\epsilon \log(\epsilon p_i) - \epsilon \log(\epsilon) - (1-\epsilon) \log(1-\epsilon)]\\
& = \Sigma_\text{i=1}^{4} -\epsilon p_i \log(p_i) + (1-\epsilon) p_i \log(1-\epsilon) 
\end{align*}
Taking the limit as $\epsilon \rightarrow 0$, the last term of the above disappears and we get the desired expression. 
\qed
\newline
Please refer to Appendix \ref{appB} for the rest of the proof that $\textbf{C$^{(p)}$}_\textbf{classical} < 1.255\epsilon+O(\epsilon^2)$. We have also seen an entanglement-assisted strategy that achieves a joint rate of $1.707\epsilon$, which must therefore be a {\em lower}-bound for the entanglement-assisted capacity. Therefore, if $\epsilon$ is chosen small enough such that the second-order terms can be ignored, we may achieve $\textbf{C$^{(p)}$}_\textbf{classical} < \textbf{C$^{(p)}$}_\textbf{quantum}$. This is suffices to prove the desired capacity separation. 
\end{proof}

The following is a corollary of Lemma \ref{lem}:

\begin{lemma}[Quantum vs. Super-quantum capacities of Channel II]
For sufficiently small $ \epsilon, \textbf{C$^{(\epsilon)}$}_\textbf{quantum} < \textbf{C$^{(\epsilon)}$}_\textbf{super-quantum}$
\end{lemma}

\begin{proof}
This statement follows straightforwardly from Equation \eqref{exp}, which gives us an expression (up to first order in $\epsilon$) for the capacity of the channel in terms of the probabilities of the $xx/ee$ and $ee$ states being sent. This expression is valid for any coding strategy, no matter what types of resources are used. 

We also know that even with entanglement, the maximum percentage of time that an $xx/ee$ state is sent is 0.8536, and this follows from CHSH (refer to the proof of Lemma \ref{quant} for why). Using our convention for defining the $p_i$s, this translates to the constraint that \begin{equation} \label{const} \Sigma_i p_i = 0.8536. \end{equation} Therefore, the entanglement-assisted capacity is upper-bounded by the maximum value of the LHS of Equation \eqref{exp}, 
\[
- \epsilon \left(p_1 \log p_1 + p_2 \log p_2 + p_3 \log p_3 
+ p_4 \log p_4 \right) + O(\epsilon^2).
\]
Under the constraint of Equation \eqref{const}, this is strictly less than 2$\epsilon$ (the super-quantum capacity). Furthermore, the bound holds even if the senders are allowed to share other types of entanglement than just $2\times n$ entanglement, since that does not affect the maximum success probability of the CHSH game (from which we derived the constraint \eqref{const}). 
\end{proof}

\section{Discussion and conclusion}
\label{sec: 6}
We have exhibited two channels that show the following new separations in classical capacity on the given classes of resources:
\begin{itemize}
\item Channel I: $\textbf{C}_\textbf{classical}, \textbf{C}_\textbf{quantum} <\textbf{C}_\textbf{super-quantum}$
\item Channel II: $\textbf{C}_\textbf{classical}<\textbf{C}_\textbf{quantum} < \textbf{C}_\textbf{super-quantum}$. 
\end{itemize}

The takeaway point from this research is that PR-boxes shared between a set of transmitters can be used for better channel communication -- a task for which they have never been considered. However, all the PR-boxes here are assumed to be perfect. We would like to see a rigorous proof that these separations can be maintained even if the senders are provided a noisy PR-box and allowed multiple uses of it for non-locality distillation. 

We have also only considered interference channels operating on discrete-variable bits because this is a proof-of-concept. In real life, many communication scenarios where multiple uncoordinated links share a common communication medium can be represented as interference channels (albeit ones where transmitted messages take on continuous values in $\mathbb{C}$ subject to Gaussian noise). Therefore, some work is needed to replicate the above separations on an general interference channel, or at the very least, characterize channels and coding strategies in a way that optimizes them for each of the three classes of resources. 

Our choice to limit our channel to handling only classical information (as opposed to density matrices representing quantum information) proved fruitful, as it paved the way for proofs that rely on classical information theory, as well as some results from pseudo-telepathy games where the referee, too, accepts only a discrete (albeit distributed) set of outcomes. In hindsight, it seems natural to draw this connection given that pseudo-telepathy games exhibit the twin boons of being {\em known} to demonstrate super-quantum-to-quantum separations, and having had winning strategies (in a few cases) characterized and generalized to an arbitrarily large number of parties \cite{broadbent, arkhipov}! It would be very satisfying if a general strategy could be found to map all pseudo-telepathy games to channels which demonstrate capacity separations for the multi-sender ($n\geq 3$) case. 

\section{Acknowledgments}
The authors thank Isaac Chuang for discussions and suggestions to improve this paper. Y.Q. was supported by a DSTA Undergraduate Scholarship (Overseas) from the Singapore government. P.W.S. was supported in part by NSF grant CCF-121-8176, and by the NSF through the Science and Technology Center for Science of Information under Grant CCF0-939370.

\appendix

\section{Remainder of proof of Lemma 1 for Channel I}
\label{app0}
Here we show that if Sender 1 uses the alphabet $\{00,01\}$ and Sender 2 uses $\{00,10\}$ (Table \ref{schem} depicts this schematically), then the second sender-receiver pair cannot communicate perfectly, and therefore $R_2 <1$ as we asserted. The same turns out to be true for the other 3 cases.

To get $I(X_1:Y_1) = H(X_1) - H(X_1|Y_1) = 1$ when there are only two options for $X_1$, the first term must take its maximal value of 1, which can only happen if $X_1$ is uniformly distributed over $\{00,01\}$. Let $X_2$ send $00$ with probability $c$ and $01$ with probability $1-c$. This is shown on the left in Table \ref{schem}. Since we will be interested in calculating $I(X_2:Y_2)$, we also calculate the input-output probability distribution experienced by sender-receiver pair 2, shown on the right in Table \ref{schem}.

\begin{table}[!htbp]
  \begin{tabular}{ c || c | c | c | c ||}
 $X_1$\textbackslash $X_2$  & \textbf{00} & \textbf{10} \\ \hline \hline
\textbf{00} & 00  & 01 \\\hline
\textbf{01} & 11 & 10  \\ \hline
\hline
  \end{tabular} \quad
  \begin{tabular}{ c || c | c | c | c ||}
 $X_2$\textbackslash $Y_2$  & \textbf{0} & \textbf{1} \\ \hline \hline
\textbf{00} & $\frac{c}{2}$  & $\frac{c}{2}$ \\\hline
\textbf{10} & $\frac{1-c}{2}$ & $\frac{1-c}{2}$  \\ \hline
\hline
  \end{tabular}
\caption{\label{schem}Left: reduced alphabets of senders and resulting output to the receivers (in the format $Y_1Y_2$). Right: Joint probability distribution experienced by the second sender-receiver pair on this coding scheme.}
  \end{table}
Referring to the right side of Table \ref{schem}, we obtain
\begin{align}
I(X_2:Y_2) &= H(X_2)+H(Y_2)-H(X_2,Y_2) \nonumber \\
&= \left[- c \log c -(1-c) \log (1-c) \right] + 1 \nonumber  \\
& - \left[ 2 \left(- \frac{c}{2} \log \frac{c}{2} \right) + 2 \left(- \frac{1-c}{2} \log \frac{1-c}{2} \right) \right] \nonumber \\
&= 0 
\end{align}

We have therefore shown that $I(X_1:Y_1) =1$ implies that $I(X_2:Y_2)=0$, so that $I(X_1:Y_1)+I(X_2:Y_2) =2$ will never be achieved. \textbf{Put another way, perfect coding between one pair implies that the other pair can do no better than random guessing.} 

\section{A classical strategy that performs as well as a hypothetical perfect entanglement-assisted strategy on Channel I}
\label{app1}
The strategy will follow after the subsequent lemmas:
\begin{lemma}\label{lemgen}
For any two-sender-receiver pair communication strategy that relies on the senders sharing some state $\ket{\Phi}$ of dimension $2\times 2$, there exists a communication strategy that achieves the same rate where the senders are restricted to sharing a state of the form $\ket{\Psi} = \alpha \ket{00} + \beta \ket{11}$, where $\alpha$ and $\beta$ are well-chosen positive real numbers. 
\end{lemma}
\begin{proof}
The key idea is to re-write $\ket{\Phi}$ in terms of its Schmidt decomposition, and then apply a unitary transformation to get $\ket{\Psi}$. Then, the senders may apply the quantum strategy whose existence we have assumed. More precisely, there exist orthogonal bases $\{\ket{A_0},\ket{A_1}\}$ for Sender 1 and $\{\ket{B_0},\ket{B_1}\}$ for Sender 2 such that $\ket{\Phi}$ can be rewritten as \[\ket{\Phi} = \alpha\ket{A_0}\ket{B_0} + \beta \ket{A_1}\ket{B_1}.\]
From there it is easy to see that Sender 1 may apply the unitary transformation $\ket{A_0}\bra{0} + \ket{A_1}\bra{1}$, and Sender 2 may apply the unitary transformation $\ket{B_0}\bra{0} + \ket{B_1}\bra{1}$, to their qubits, to transform $\ket{\Psi}$ into $\ket{\Phi}$. Any such unitary $U$ is completely accounted for in our model of communication in \ref{mmodel} by applying it to the POVMs $M_i$ that the senders choose for their states (which preserves its POVM properties), that is, using the property $U\ket{\Phi} =\ket{\Psi} \rightarrow \bra{\Phi} M_i \ket{\Phi} = \bra{\Phi} U M_i U^\dagger \ket{\Phi}.$
\end{proof}

Since the following two lemmas are almost identical to the ones in \cite{BMT}, we merely cite them and leave the reader to refer to \cite{BMT} for their proofs.
\begin{lemma}
For any two-party quantum communication protocol that uses an entangled state of dimension $d_A \times d_B$, there exists a two-party quantum communication protocol that uses a state of dimension $d \times d$ where $d:=\min(d_A,d_B)$. 
\end{lemma}
This justifies the audaciously general claim made in Lemma \ref{lemquan} that {\em no} quantum state of dimension $2 \times d$ could possibly enable a perfect joint rate for communication. The proof is similar to the proof of Lemma \ref{lemgen} and relies on the following fact from the Schmidt decomposition: if $H_1$ and $H_2$ are Hilbert spaces of dimensions $n,m$ respectively, and we assume without loss of generality that $n\geq m$, for any vector $w \in H_1\otimes H_2$, there exist orthonormal bases $\{u_i, 1\leq i \leq n\}$ for $H_1$ and $\{v_j, 1\leq j \leq m\}$ for $H_2$ respectively such that 
\begin{equation}
w = \Sigma_{i=1}^m \alpha_i u_i \otimes v_i.
\end{equation}

\begin{lemma}
Any POVM can be written in a way such that all its elements are proportional to one-dimensional projectors. Each such projector can be re-written in the form 
\begin{align}
P = \begin{pmatrix}\cos^2(\theta) & e^{-i\phi} \sin(\theta)\cos(\theta) \\e^{i\phi} \sin(\theta)\cos(\theta)& \sin^2(\theta)\end{pmatrix} \label{eq:proj}
\end{align}
for appropriate angles $0 \leq \theta \leq \frac{\pi}{2}$ and $0 \leq \phi \leq 2\pi$. Since this representation is unique, we may associate each such projector with a three-dimensional unit vector $\vec{v} = (\sin(2\theta)\cos(\phi),\sin(2\theta)\sin(\phi), \cos(2\theta))$.
\end{lemma}

Finally, the classical strategy promised three lemmas ago is described. Thanks to Lemma \ref{lemgen}, we may assume that the two senders are using an entangled state of the form $\ket{\Psi} = \alpha \ket{00} + \beta \ket{11}$, where $\alpha$ and $\beta$ are strictly positive real numbers. 

Suppose a quantum strategy exists and the POVMs applied by the two senders, $M^x := {\mathcal X} (x) = \{\gamma_i^xP_i^x\} $ and $N^y := {\mathcal Y}(y) = \{\gamma_j^yQ_j^y\}$ have been fixed beforehand for each $x,y \in \{0,1\}$. We will show that any measurement outcome $(i,j)$ on $\ket{\Psi}$ as described in the first row of Equations ~\ref{mmodel} can be replicated perfectly classically. The probability of getting the tuple $(i,j)$ is:
\begin{align}
\Pr[i,j] & = \bra{\Psi} (\gamma_i^x P_i^x) \otimes (\gamma_j^y Q_j^y) \ket{\Psi} \nonumber \\
& = \gamma_i^x \gamma_j^x \left[\alpha^2\cos^2(\theta_i^x) \cos^2(\theta_j^y) \right.\nonumber \\
&\left. + 2\alpha\beta \left[\cos(\phi_i^x + \phi_j^y) \sin \theta_i^x \cos \theta_i^x \sin \theta_j^y \cos \theta_j^y\right] \right. \nonumber\\
& \left. + \beta^2 \sin^2 (\theta_i^x) \sin^2(\theta_j^y) \right] \nonumber\\
& = \gamma_i^x \gamma_j^x (a^2 + b^2 +2 abc) \label{eq:BMT}
\end{align}
where $a:= \alpha \cos(\theta_i^x) \cos(\theta_j^y)$, $b:= \beta \sin(\theta_i^x) \sin(\theta_j^y)$ and $c:= \cos(\phi_i^x + \phi_j^y)$. 
Using the AM-GM inequality and the fact that $|c|\leq 1$ we may show that $\Pr[i,j]$ can only vanish if one of the following two things are true of the POVMs used by the two senders ($\{\gamma_i^xP_i^x\}, \{\gamma_j^yQ_j^y\}$):
\begin{itemize}
\item $a=b=0$ \\
Attained if $\theta_i^x=0 \, , \theta_j^y = \pi/2$ or vice versa -- that is, either $P_i^x$ or $Q_j^y$ belongs to neither hemisphere. 
\item $a=b$ and $c=-1$. \\
Attained if $\phi_i^x+\phi_j^y = \pi$(both projectors in eastern hemisphere) or $\phi_i^x+\phi_j^y = 3\pi$ (both projectors in western hemisphere). 
\end{itemize} 
But all our classical strategy needs to do is to choose a classical tuple, $(i,j)$, such that the corresponding quantum POVM elements, $P_i^x$ and $Q_j^y$, would not fulfill either of these conditions. To do this, it suffices for Sender 1, knowing $M^x :=  \{\gamma_i^xP_i^x\} $, to choose an $i$ such that $P_i^x$ belongs to the eastern hemisphere and for Sender 2, knowing $N^y := \{\gamma_j^yQ_j^y\}$, to choose a $j$ such that $Q_j^y$ belongs to the western hemisphere (without actually measuring anything). This is always possible since POVM elements have to sum to the identity. They may then carry out the (classical) mappings ${\mathcal A}$ and ${\mathcal B}$ on their message bits and POVM `outcomes' as per normal. 

\section{A proof of an upper bound on the classical capacity of Channel II}
\label{appB}

Channel II has been replicated in Table \ref{50} for your convenience. We would like to prove that for some values of the parameter $\epsilon$, the entanglement-assisted capacity beats the classical capacity. 

\begin{table}[!htbp]
\centering
  \begin{tabular}{ c || c | c | c | c ||}
 $X_1$\textbackslash $X_2$  & \textbf{00} & \textbf{01} & \textbf{10} & \textbf{11} \\ \hline \hline
\textbf{00} & 00/ee   &  ee & 01/ee  & ee \\ \hline
\textbf{01} & ee  & 00/ee  & ee & 01/ee \\ \hline
\textbf{10} & 10/ee & ee & ee & 11/ee  \\ \hline
\textbf{11} & ee &  10/ee & 11/ee  & ee \\ \hline
\hline
  \end{tabular}
\caption{\label{50}Channel II, reproduced here. In bold are the senders' inputs, and the table shows the resulting channel outputs. Erased bits are denoted by $e$, and the $ee$ has probability $1-\epsilon$ in the squares it shares with numerical values. }
\end{table}

Let the probability of $X_1$ sending $00$ or $01$ be $p$, and the probability of $X_2$ sending $00$ or $01$ be $q$. 

The proof proceeds in three steps. We aim to show that the parameter $\epsilon$ governing the rate for the best classical strategy can be tuned small enough that that quantum-assisted capacity is larger than the classical capacity. Therefore, we first establish a relation that constrains the probabilities of the various possible output symbols for any classical strategy. Next, we find an expression for the classical capacity of the channel up to first order in $\epsilon$. Using this relation, we find an upper bound on the classical capacity in terms of $\epsilon$. This completes the proof.  

The first thing to prove is that 
\begin{lemma}
There are numbers
$\alpha$, $\beta$, $\gamma$ and $\delta$ with
\[
\alpha + \beta + \gamma + \delta \leq 3 \qquad \text{and} \qquad \alpha, \beta,\gamma,\delta < 1
\]
such that if we look at the output,
\begin{eqnarray} \label{eq:0}
\Pr(00) &=& \epsilon p q \alpha, \nonumber \\
\Pr(01) &=& \epsilon p (1-q) \beta, \nonumber \\
\Pr(10) &=& \epsilon (1-p) q \gamma, \nonumber \\
\Pr(11) &=& \epsilon (1-p) (1-q) \delta .
\end{eqnarray}
\end{lemma}

{\bf Proof:}
First, let's assume that Alice and Bob input a product distribution. The
most general thing they can do is input a convex combination of product
distributions, and the result for convex combinations follows
straightforwardly from the result for product distributions.

Now, let Alice's input be expressed as a vector.
\[
(p_{00}, p_{01}, p_{10}, p_{11})
\]
meaning that with probability $p_{ij}$ Alice inputs bit string $ij$.
Note that $p_{00} + p_{01} = p$ and $p_{10} + p_{11} = 1-p$.
We can decompose this vector into a sum of `basis vectors' $\{u_i\}$, each with two non-zero entries $a_i$ and $b_i$. We may stipulate $a_i/b_i = p/(1-p)$:
\begin{align*}
&(p_{00}, p_{01}, p_{10}, p_{11}) = \Sigma_{i=1}^4 u_i \\
&=(a_1,0,b_1,0) + (a_2, 0 , 0 ,b_2) + (0,a_3,b_3,0) \\
&+ (0, a_4, 0, b_4). 
\end{align*}

Similarly, we decompose Bob's input distribution into $\{v_i\}$ such that $c_i/d_i = q/(1-q)$.  
\begin{align*}
&(q_{00}, q_{01}, q_{10}, q_{11})= \Sigma_{i=1}^4 v_i \\
&=(c_1,0,d_1,0) + (c_2, 0 , 0 ,d_2) + (0, c_3, d_3, 0) \\
&+ (0, c_4, 0, d_4). 
\end{align*}

This notation permits the senders' joint inputs to be written as a linear combination of 16 terms, ${\Sigma}_{i,j} u_i v_j$. Each such term induces a probability distribution over channel outputs, which we shall express using the same naming convention for the proportionality factors as in Equation \eqref{eq:0}, but with an additional subscript $i,j$. We claim that for each basis vector of the joint input distribution (indexed by $i,j$), $\alpha_{ij} + \beta_{ij} + \gamma_{ij} + \delta_{ij} \leq 3$. 

For instance, if we take the vectors $ u_2 = (a_2,0,0,b_2)$ and $v_1 = (c_1, 0, d_1, 0)$, we know that $\Pr(u_2) = a_2+b_2$ and  $a_2 = p\Pr(u_2)$, $b_2 = (1-p)\Pr(u_2)$ by our convention for choosing the entries of the basis vectors. Similarly, $c_1 = q\Pr(u_2,v_1) $ and $d_1 = (1-q)\Pr(u_2,v_1) $. Then we have 

\begin{alignat*}{2}
{\Pr}_\text{21}(00) = &\quad \epsilon a_2 c_1 = \epsilon p q \alpha_{21} \Pr(u_2,v_1) \\
{\Pr}_\text{21}(01) = &\quad \epsilon a_2 d_1 = \epsilon p (1-q) \beta_{21} \Pr(u_2,v_1) \\
{\Pr}_\text{21}(10) = &\quad 0  \quad = \epsilon (1-p) q \gamma_{21} \Pr(u_2,v_1)\\
{\Pr}_\text{21}(11) = &\quad \epsilon b_2 d_1 = \epsilon (1-p) (1-q) \delta_{21} \Pr(u_2,v_1)\\
\end{alignat*}

where each $\alpha_{ij},\beta_{ij},\gamma_{ij},\delta_{ij}$ is $1$ if the corresponding probability is non-zero and $0$ otherwise. In this instance, $\alpha_{21} = \beta_{21} = \delta_{21} = 1$, $\gamma_{21} = 0$. In particular, $\alpha_{21}+\beta_{21}+\gamma_{21}+\delta_{21} \leq 3$ and we can easily check that this is true for all choices of $i,j$. It is straightforward to extend this property to $\alpha+\beta+\gamma+\delta$. We have \[\alpha = \Sigma_{ij} \Pr(u_i,v_j) \alpha_{ij}\] and so on for the other greek letters. So 
\begin{align*}
&\alpha+\beta+\gamma+\delta \\
&= \Sigma_{ij} \Pr(u_i,v_j) (\alpha_{ij} + \beta_{ij}+\gamma_{ij} +\delta_{ij})  \\
& \leq 3 \left(\Sigma_{ij} \Pr(u_i,v_j) \right)= 3
\end{align*}
\qed

The next step is to find an upper bound to the classical capacity of this channel up to first order in $\epsilon$. It follows from Lemma \ref{lem} a channel with these probabilities cannot send much more than
\begin{align}
\begin{split}
- \epsilon  &\left[ pq \alpha \log( pq \alpha) + p(1-q) \beta \log( p(1-q) \beta) \right. \\
& \left. + (1-p)q \gamma \log((1-p)q \gamma) \right.\\
& \left. + (1-p)(1-q) \delta \log((1-p)(1-q) \delta) \right]
\end{split}
\end{align}
information. Now, we relax the problem. We no longer require that we have a product distribution. Choose $p_i$ as described in Equations ~\ref{sub} and choose $\alpha$, $\beta$, $\gamma$, $\delta$ with $\alpha+\beta+\gamma+\delta \leq 3$ as in Equation ~\ref{eq:0}. The capacity of our channel, by the above lemma, is at most
\begin{align} \label{eq:2}
&- \epsilon (\alpha k_{00} \log \alpha k_{00} + \beta k_{01} \log \beta k_{01}  \nonumber \\
&+ \gamma k_{10} \log \gamma k_{10} + \delta k_{11} \log \delta k_{11} ) +O(\epsilon^2). 
\end{align}
where we have further defined $k_{00} = pq, k_{01} = p(1-q),  k_{10} = (1-p)q, k_{11} = (1-p)(1-q)$ such that $ \Sigma_{i,j} k_{ij} = 1$. 

Our aim now is to find values of $(k_{00}, k_{01},k_{10}, k_{11})$ and $(\alpha,\beta,\gamma,\delta)$ which maximize this expression, which would give us an upper bound on the classical capacity.
 
We can do this in several steps, which we outline below. 

First, we observe that one of $\alpha k_{00}$, $\beta k_{01}$, $\gamma k_{10}$, $\delta k_{11}$, is at most $3/16$, and recall that $\alpha, \beta,\gamma,\delta < 1$. But $f(x) = x \log x$ is maximized when $x = 1/e$. These two facts let us assume that \[\alpha + \beta + \gamma + \delta = 3\] at the point where Equation \eqref{eq:2} is maximized because of the following: suppose $\alpha k_{00} < \frac{3}{16}$ (and therefore < $\frac{1}{e}$). Then if $\alpha+ \beta+ \gamma+ \delta < 3$, we could increase the capacity, Equation \eqref{eq:2}, by increasing $\alpha$, and therefore our original choice could not have maximized the capacity.

Next, we show that at the maximum
$\alpha = 3 k_{00}$, 
$\beta = 3 k_{01}$, 
$\gamma = 3 k_{10}$, 
$\delta = 3 k_{11}$. 

{\bf Proof:}
This eventually falls out from formulating the problem with Lagrange multipliers with the constraints $\alpha + \beta + \gamma + \delta = 3$ and $k_{00}+k_{01}+k_{10}+k_{11} =1$. But we take a quicker tack: we show that we can increase the rate if this is not the case. 

Suppose $\alpha - 3 k_{00} = \delta_1$ and $\beta - 3k_{01} = -\delta_2$. Let $\epsilon$ be $\frac{1}{2} \min(\delta_1,\delta_2)$. We can increase $\alpha k_{00}$ and $\beta k_{01}$ by replacing
\begin{eqnarray*}
\alpha' &=& \alpha - \epsilon \\
\beta' &=& \beta + \epsilon \\
k_{00}' &=& k_{00} + \epsilon/3\\
k_{01}' &=& k_{00} - \epsilon/3
\end{eqnarray*}
which respects the constraints while leaving the other four variables unchanged. Since the probability of getting a faithfully-transmitted output increases with both $\alpha k_{00}$ and $\beta k_{01}$ (recalling how $k_{ij}$ was defined), so should the rate increase. 
\qed

Finally, we need to show that the capacity is maximized when
$k_{00} = k_{01} = k_{10} = k_{11} = 1/4$ and  
$\alpha = \beta = \gamma = \delta = 3/4$.

{\bf Proof:}
Let
\[
f = -3x^2 \log 3 x^2.
\]
We need to find
\begin{align*}
\max_{k_{ij}} \, f(k_{00}) + f(k_{01}) + f(k_{10}) + f(k_{11})  \, \\
s.t. \,\, \Sigma_{i,j} \, k_{ij} = 1
\end{align*}
Define
\begin{align*}
{\mathcal{L}}(k_{ij},\lambda)  &= f(k_{00}) + f(k_{01}) + f(k_{10}) + f(k_{11})   \\
&+ \lambda(k_{00} + k_{01} + k_{10} + k_{11} - 1)
\end{align*}
and we would like to get $\nabla_{ij,\lambda} {\mathcal{L}}= 0,$  so we need $f'(k_{00}) =f'(k_{01}) =f'(k_{10}) =f'(k_{11}) $.
 
Looking at the graph of $f'$ shows that for any $t$, 
there are at most two points $x_1$ and $x_2$ with $0 < x_1 \leq x_2 < 1$ 
where $f'(x_1) = f'(x_2) = t$. This shows that in the maximum of
$ f(k_{00}) + f(k_{01}) + f(k_{10}) + f(k_{11}) $, there are at most two different values
of $k_{ij}$. However, the asymmetry of the graph for the portions where $0<x<1$ and $f'(x)>0$ makes it clear that one cannot assign more than one value to $k_{ij}$in such a way that $\Sigma k_{ij}=1$. 
We may thus conclude that there is only one value of $k_{ij}$ that maximizes this expression, and that must be $k_{ij} = 1/4$. 

Putting the numbers into Equation \eqref{eq:2}, we may conclude that the upper bound on the classical capacity is $1.255\epsilon + O(\epsilon^2)$.  

\bibliography{main}
\bibliographystyle{ieeetr}
\nocite{*}

\end{document}